\documentclass[11pt, a4paper]{article}

\usepackage[margin = 1in]{geometry}
\usepackage{palatino}

\usepackage{amsmath}
\usepackage{amssymb}
\usepackage{amsfonts}
\usepackage{amsthm} 
\usepackage{graphicx}
\usepackage{comment}
\usepackage{thm-restate}
\usepackage{hyperref}
\usepackage{here}
\usepackage{cite}
\usepackage{enumerate}
\usepackage{tabularx}
\usepackage{environ}
\usepackage{url}
\usepackage[normalem]{ulem}
\usepackage{color}
\usepackage[T1]{fontenc}
\usepackage{enumitem}

\newtheorem{theorem}{Theorem}
\newtheorem{lemma}[theorem]{Lemma}
\newtheorem{proposition}[theorem]{Proposition}

\newtheorem{observation}[theorem]{Observation}

\usepackage[linesnumbered,noend,ruled,vlined]{algorithm2e}
\SetKwInput{KwInput}{Input}     % Set the Input
\SetKwInput{KwOutput}{Output}   % set the Output
\SetKw{KwTo}{to}

\usepackage[framemethod=tikz]{mdframed}

\global\long\def\I{\mathcal{I}}%
\global\long\def\tO{\widetilde{O}}%
\global\long\def\eps{\varepsilon}%
\global\long\def\epsilon{\varepsilon}%
\global\long\def\tO{\tilde{O}}%
\global\long\def\eps{\varepsilon}%

\newcommand{\Span}{\mathrm{span}}
\newcommand{\Rank}{\mathrm{rank}}

\newcommand{\M}{\mathbf{M}}
\newcommand{\cI}{\mathcal{I}}
\newcommand{\cB}{\mathcal{B}}
\newcommand{\OPT}{\mathsf{OPT}}
\newcommand{\R}{\mathbb{R}}
\newcommand{\Z}{\mathbb{Z}}
\newcommand{\rup}[1]{\lceil #1 \rceil_{1 + \eps}}
\newcommand{\bp}{\bar{p}}

\DeclareMathOperator{\spn}{span}
\newcommand{\optw}{\mathrm{opt}}

\newcommand{\ShowComment}{true}
\ifdefined\ShowComment
\newcommand{\tatsuya}[1]{{\bf \color{red} tatsuya: #1}}
\newcommand{\later}[1]{{\color{red} [to be revised] #1}} % text to be revised later (shown in red)
\else
\newcommand{\tatsuya}[1]{}
\newcommand{\later}[1]{}
\fi

\title{A Faster Auction Algorithm \\ for Weighted Matroid Intersection}

\author{Tatsuya Terao\thanks{
Research Institute for Mathematical Sciences, Kyoto University.
E-mail: ttatsuya@kurims.kyoto-u.ac.jp
}
}

\date{}

\begin{document}

\maketitle

\begin{abstract}
We consider the weighted matroid intersection problem in the independence-oracle model. A sequence of works by Huang--Kakimura--Kamiyama~[SODA'16 \& Math. Program'19], Chekuri--Quanrud~[SODA'16], Quanrud~[ICALP'24], and Dudeja--Grilnberger~[IPCO'26] has developed efficient $(1-\varepsilon)$-approximation algorithms for this problem.

We present a simple deterministic auction algorithm that, given two matroids on a common ground set of size $n$, computes a $(1-\varepsilon)$-approximate maximum-weight common independent set using $O(n \varepsilon^{-2} \log^2(n))$ independence-oracle queries. This is the first deterministic $(1-\varepsilon)$-approximation algorithm for the weighted matroid intersection problem whose query complexity is nearly linear in $n$ and polynomial in $1/\varepsilon$. Our algorithm builds on the auction algorithm for unweighted matroid intersection by Huang--Kobayashi~['26].
\end{abstract}

\section{Introduction} \label{sec:intro}

In the weighted matroid intersection problem, we are give two matroids $\M_1 = (E, \I_1)$ and $\M_2 = (E, \I_2)$ on a common ground set $E$, together with a weight function $w: E \to \mathbb{R}_{> 0}$.
The goal is to find a common independent set $S \in \I_1 \cap \I_2$ that maximizes $w(S) := \sum_{e \in S} w(e)$.
Matroid intersection is one of the fundamental problems in combinatorial optimization and generalizes several classical optimization problems, including bipartite matching. The design of faster algorithms for matroid intersection has received considerable attention in recent years.
A particularly important objective is to reduce the number of independence-oracle queries. 
An independence oracle for a matroid $\M=(E,\I)$ takes a set $S\subseteq E$ as input and determines whether $S\in\I$.

\subsection{Previous Work}

Starting the work of Edmonds~\cite{edmonds1970submodular, edmonds1979matroid}, a long line of research has developed increasingly efficient algorithms for matroid intersection with polynomial query complexity~\cite{aigner1971matching, lawler1975matroid, cunningham1986improved, brezovec1986two, fujishige1995efficient, shigeno1995dual, lee2015faster, chekuri2016fast, huang2016exact, nguyen2019note, chakrabarty2019faster, blikstad2021breaking_STOC, blikstad2021breaking, u2022subquadratic, blikstad2023fast,quanrud:LIPIcs.ICALP.2024.118,blikstad2024efficient,terao:LIPIcs.WADS.2025.50,dudeja2026weighted}.

\subparagraph*{Unweighted Matroid Intersection.}

In the unweighted matroid intersection problem, the goal is to find a common independent set of maximum size.
Edmonds~\cite{edmonds1970submodular} developed the first polynomial-query algorithm by reduction to the matroid partitioning problem. This algorithm requires $O(n^4)$ independence oracle queries.
More direct algorithms were given by Aigner--Dowling \cite{aigner1971matching} and Lawler \cite{lawler1975matroid}.
Lawler's algorithm requires $O(nr^2)$ independence oracle queries, where $r$ denotes the maximum cardinality of a common independent set.
In 1986, Cunningham \cite{cunningham1986improved} presented an $O(nr^{3/2})$-independence-query algorithm by using a blocking flow approach, which is akin to the bipartite matching algorithm of Hopcroft--Karp \cite{hopcroft1973n}.
Chekuri--Quanrud \cite{chekuri2016fast} pointed out that an $O(nr/\eps)$-independence-query $(1 - \eps)$-approximation algorithm can be obtained by terminating Cunningham's algorithm early.
Recently, Nguy$\tilde{{\hat{\text{e}}}}$n~\cite{nguyen2019note} and Chakrabarty--Lee--Sidford--Singla--Wong~\cite{chakrabarty2019faster} independently presented a new binary search technique and developed a combinatorial $\tilde{O}(n r)$-independence-query exact algorithm.\footnote{The $\tO$ notation omits factors polynomial in $\log n$.}
Chakrabarty--Lee--Sidford--Singla--Wong also presented a new augmenting sets technique and developed an $\tO(n^{1.5}/\epsilon^{1.5})$-independence-query $(1 - \epsilon)$-approximation algorithm.
Blikstad--van den Brand-Mukhopadhyay--Nanongkai~\cite{blikstad2021breaking_STOC} first broke the $\tilde{O}(n^2)$-independence-query bound for exact matroid intersection algorithms.
Blikstad~\cite{blikstad2021breaking} improved the independence query complexity of a $(1 - \epsilon)$-approximation algorithm to $\tilde{O}(n \sqrt{r} / \eps)$.
Blikstad's improvement on the $(1 - \epsilon)$-approximation algorithm resulted in a randomized $\tilde{O}(nr^{3/4})$-independence-query exact algorithm and a deterministic $\tilde{O}(nr^{5/6})$-independence-query exact algorithm.
Recently, Quanrud \cite{quanrud:LIPIcs.ICALP.2024.118} presented a randomized $\tilde{O}_{\eps}(n + r^{3/2})$-independence-query $(1 - \epsilon)$-approximation algorithm\footnote{The $\tilde{O}_{\eps}$ notation omits factors polynomial in $\eps$ and $\log n$.}.
Remarkably, Blikstad--Tu~\cite{blikstad2024efficient} recently presented a randomized $\tilde{O}_{\eps}(n)$-independence-query $(1 - \epsilon)$-approximation algorithm.
Very recently, Huang--Kobayashi~\cite{huang2026auction} presented an elegant deterministic $\tilde{O}_{\eps}(n)$-independence-query $(1 - \epsilon)$-approximation algorithm.

\subparagraph*{Weighted Matroid Intersection.}
For the weighted matroid intersection problem, Edmonds~\cite{edmonds1970submodular} also gave the first polynomial-query algorithm.
Faster exact algorithms were developed in \cite{lawler1975matroid, brezovec1986two, frank1981weighted, fujishige1995efficient, shigeno1995dual, lee2015faster}.
Using a cutting-plane method, Lee--Sidford--Wong~\cite{lee2015faster} presented an $\tO(n^2)$-independence-query exact algorithm.

Several recent works have studied efficient $(1-\varepsilon)$-approximation algorithms for weighted matroid intersection.
Huang--Kakimura--Kamiyama~\cite{huang2016exact} presented a deterministic $\tilde{O}_{\eps}(nr^{3/2})$-independence-query $(1 - \epsilon)$-approximation algorithm.
Later, Chekuri--Quanrud~\cite{chekuri2016fast} improved the query complexity to $\tilde{O}_{\eps}(nr)$.
Subsequently, Quanrud~\cite{quanrud:LIPIcs.ICALP.2024.118} introduced a new sparsification technique and obtained a randomized $\tilde{O}_{\eps}(n + r^{3/2})$-independence-query $(1 - \epsilon)$-approximation algorithm.
Very recently, Dudeja--Grilnberger~\cite{dudeja2026weighted} gave a new reduction that transforms any $\alpha$-approximation algorithm for unweighted matroid intersection into an $\alpha(1-\eps)$-approximation algorithm for the weighted problem.
Their reduction incurs an $\varepsilon^{-O(1/\varepsilon)}$ overhead and an additional $O(\log C)$ factor, where $C:= \frac{\max_{e \in E} w(e)}{\min_{e \in E} w(e)}$ is the aspect ratio of the weights. Consequently, combining their reduction with the unweighted matroid intersection algorithms of Blikstad--Tu~\cite{blikstad2024efficient} and Huang--Kobayashi~\cite{huang2026auction} yields randomized and deterministic $(1-\varepsilon)$-approximation algorithms for weighted matroid intersection using $\tilde{O}(n \cdot (1/\eps)^{1/\eps} \cdot \log(C) )$ independence-oracle queries, respectively.

\subparagraph*{Algorithms in the Rank-Oracle Model}

Some previous work studies matroid intersection in the rank-oracle model~\cite{lee2015faster, chakrabarty2019faster, u2022subquadratic, huang2026auction}.
A rank oracle returns the rank $\Rank_{\mathcal M}(S)$ of any subset $S \subseteq E$. This model is stronger than the independence-oracle model, since the independence of $S$ can be determined by checking whether $\Rank_{\mathcal M}(S)=|S|$.

For unweighted matroid intersection, Chakrabarty--Lee--Sidford--Singla--Wong~\cite{chakrabarty2019faster} presented a deterministic $\tilde{O}_{\eps}(n)$-rank-query $(1 - \epsilon)$-approximation algorithm.
% Combined with the weighted-to-unweighted reduction of Chekuri--Quanrud~\cite{chekuri2016}, their result yields a deterministic $\widetilde{O}{\varepsilon}(n)$-rank-query $(1-\varepsilon)$-approximation algorithm for weighted matroid intersection\footnote{Although the reduction produces unweighted matroid intersection over weight-induced matroids, the binary-search technique still applies, allowing us to use the algorithm of Chakrabarty--Lee--Sidford--Singla--Wong~\cite{chakrabarty2019faster}.}.
For weighted matroid intersection, Huang--Kobayashi~\cite{huang2026auction} also presented a deterministic $\tilde{O}_{\eps}(n)$-rank-query $(1 - \epsilon)$-approximation algorithm.

\subsection{Our Result}

Our contribution is to present a simple auction algorithm for weighted matroid intersection, which yields the following result.

\begin{theorem} \label{thm:main}
    Let $\eps > 0$.
    For the weighted matroid intersection problem, there exists a deterministic $(1 - \eps)$-approximation algorithm using $O(n \eps^{-2} \log^2(n))$ independence-oracle queries. 
\end{theorem}

Our algorithm is based on the auction algorithm of Huang--Kobayashi~\cite{huang2026auction} for unweighted matroid intersection.
The resulting algorithm is the first deterministic $(1-\varepsilon)$-approximation algorithm for the weighted matroid intersection problem whose query complexity is nearly linear in $n$ and polynomial in $1/\varepsilon$.

Both the unweighted matroid intersection algorithms of Blikstad--Tu~\cite{blikstad2024efficient} and Huang--Kobayashi~\cite{huang2026auction} are inspired by the auction algorithm for the unweighted bipartite matching problem by Assadi--Liu--Tarjan~\cite{assadi2021auction}.
This naturally raises the question of whether these auction-based approaches can be extended to weighted matroid intersection while retaining nearly-linear independence-oracle query complexity. Prior to our work, it was not known whether either of these auction algorithms for unweighted matroid intersection could be directly extended to obtain a $(1-\varepsilon)$-approximation for weighted matroid intersection using a nearly linear number of independence-oracle queries.
Furthermore, Dudeja--Grilnberger~\cite{dudeja2026weighted} identified improving the dependence on $\varepsilon$ from $(1/\varepsilon)^{1/\varepsilon}$ to $\operatorname{poly}(1/\varepsilon)$ as an interesting direction.

% , as well as by the recent auction algorithm for weighted bipartite matching by Liu--Ke--Khuller~\cite{liu_et_al:LIPIcs.APPROX/RANDOM.2023.28}.

% \subparagraph*{Streaming-Model Implementation.}
% We note that our weighted matroid intersection algorithm can be implemented in the streaming model using $O(\eps^{-3}\log(1/\eps))$ passes and $O(r\log(r/\eps))$ space, where $r$ denotes the maximum cardinality of a common independent set. 
% Here, we measure space complexity in machine words and assume that each machine word contains $O(\log n)$ bits.
% This can be achieved by slightly modifying the analysis and using the weighted-to-unweighted reduction of Dudeja--Grilnberger~\cite{dudeja2026weighted}.

\section{Preliminary}\label{sec:prelim}

\subsection{Notation} \label{subsec:matroids}

Here, we provide the notation and conventions used throughout the paper.

\paragraph*{Basic Notation.} \label{subsec:basic_notation}

Let $\R$ and $\R_+$ denote the sets of real numbers and non-negative real numbers, and let $\Z$ denote the set of integers.
For a set $X$ and an element $e$, we write $X+e:=X\cup\{e\}$ and $X-e:=X\setminus\{e\}$.
For two sets $A$ and $B$, we will also write $A + B := A \cup B$ and $A - B := A \setminus B$ when no confusion can arise.

% For a set $X \subseteq E$ and a function $f: E \to \mathbb{R}$, we denote $\sum_{e \in X}f(e)$ by $f(X)$.
Throughout the paper, \(\log\) denotes the natural logarithm.

\paragraph*{Matroid.}
A pair $\M = (E, \mathcal{I})$ of a finite set $E$ and a non-empty set family $\mathcal{I} \subseteq 2^{E}$ is called a {\em matroid} if the following properties are satisfied.

\begin{description}
    \item[(Downward closure property)] 
    If $S \in \mathcal{I}$ and $S' \subseteq S$, then $S' \in \mathcal{I}$.
    \item[(Augmentation property)]
    If $S, S' \in \mathcal{I}$ and $|S'| < |S|$, then there exists $v \in S \setminus S'$ such that $S' + v \in \mathcal{I}$.
\end{description}
    
A set $S \subseteq E$ is called {\em independent} if $S \in \mathcal{I}$ and {\em dependent} otherwise.
An element $e$ is called a {\em loop} if $\{ e \} \notin \I$.
An inclusion-wise maximal independent set is called a {\em base} of $\M$.
It is well known that every base of $\M$ has the same size.
It is also well known that, for any two bases $B_1$ and $B_2$ of $\M$ and any element $e\in B_1\setminus B_2$, there exists an element $f\in B_2\setminus B_1$ such that both $B_1-e+f$ and $B_2+e-f$ are bases of $\M$; see e.g. \cite[Theorem 39.12]{schrijver2003combinatorial}.
This is known as the \emph{simultaneous exchange property}.
An inclusion-wise minimal dependent set is called a {\em circuit} of $\M$.

For a matroid $\M = (E, \I)$, we define the {\em rank function} $\Rank_{\M}: 2^E \to \mathbb{Z}_{\geq 0}$ and the {\em span function} $\Span_{\M}: 2^E \to 2^E$ by
\begin{align*}
& \Rank_{\M}(X) = \max \{ |Y| \mid Y \subseteq X, Y \in \I \}, \\
& \Span_{\M}(X) = \{ e \in E \mid \Rank_{\M}(X + e) = \Rank_{\M}(X) \}
\end{align*}
for any $X \subseteq E$.

For an independent set $X \in \I$ and an element $e \in E \setminus X$ such that $X + e \notin \I$, $X + e$ contains a unique circuit, denoted by $C_{\M}(X, e)$, called the {\em fundamental circuit} $e$ with respect of $X$.

For $X\subseteq E$ and $q\in\R^E$, we write $q(X):=\sum_{e\in X}q(e)$.
A \emph{$q$-maximum base} of $\M$ is a base $B$ maximizing $q(B)$.
The notions of $q$-maximum independent sets and $p$-minimum bases (for $p\in\R^E$) are defined similarly.
Such sets can be computed by the greedy algorithm (see, e.g., \cite[Section~40.1]{schrijver2003combinatorial}).

For $k\in\Z$ with $k\ge0$, the \emph{truncation} of $\M$ to rank $k$ is the matroid $(E,\{X\in\cI\mid |X|\le k\})$.
For $Y\subseteq E$, the \emph{contraction} of $\M$ by $Y$ is the matroid $\M/Y:=(E\setminus Y,\{X\subseteq E\setminus Y\mid X\cup I_Y\in\cI\})$, where $I_Y$ is a maximal independent subset of $Y$.
This definition does not depend on the choice of $I_Y$.

We assume that each matroid $\M=(E,\I)$ is given by an \emph{independence oracle}, which, given a subset $S\subseteq E$, determines whether $S\in\I$.

\paragraph*{Weighted Matroid Intersection.}
Given two matroids $\mathcal{M}_1 = (E, \mathcal{I}_1)$ and $\mathcal{M}_2 = (E, \mathcal{I}_2)$ defined on the same ground set $E$, a {\em common independent set} $S$ is a set in $\I_1 \cap \I_2$.
Let \(n:=|E|\), and let $w\in \R_+^E$ be a nonnegative weight function on $E$.
For $S\subseteq E$, let $w(S):=\sum_{e\in S}w(e)$.
The {\em weighted matroid intersection problem} asks for a common independent set $S$ maximizing $w(S)$.

The main focus of this paper is to present an efficient $(1-\eps)$-approximation algorithm for the weighted matroid intersection problem.
Let $\OPT_w$ denote a maximum-weight common independent set.
For a given $\eps\in(0,1)$, the goal is to find a common independent set of weight at least $(1 - \eps) \cdot w(\OPT_w)$.

For $i \in \{1, 2\}$, the rank function $\Rank_{\M_i}(\cdot)$, the span function $\Span_{\M_i}(\cdot)$, and the fundamental circuit $C_{\M_i}(\cdot)$ are denoted by $r_i(\cdot)$, $\Span_i(\cdot)$, and $C_i(\cdot)$, respectively.

As in~\cite[Section~2.1]{huang2026auction}, we assume without loss of generality the following.
\begin{itemize}[leftmargin=*]
\item $E\neq\emptyset$, and $\M_1$ and $\M_2$ have no loops. In particular, $w(e)\le w(\OPT_w)$ for every $e\in E$.
\item $r_1(E)\le 2r$. Indeed, we can replace $\M_1$ by its truncation to rank $2|X_0|$, where $X_0$ is a maximal common independent set found greedily with $O(n)$ queries.
Since $|X_0|\ge r/2$, every common independent set $S$ satisfies $|S|\le r\le 2|X_0|$, and hence $S$ is also independent in the truncated matroid.
\end{itemize}

\subsection{Basic Properties}

In this subsection, we present basic properties of matroids from Huang--Kobayashi~\cite{huang2026auction} that are used throughout the paper.

\begin{observation}[{\cite[Observation~2.1]{huang2026auction}}]\label{obs:circuit}
Let $\M=(E,\cI)$ be a matroid, let $X\in\cI$, and let $e\in E\setminus X$ with $X+e\notin\cI$.
Then, for $f\in X+e$, we have $X+e-f\in\cI$ if and only if $f\in C_\M(X\mid e)$.
\end{observation}

\begin{lemma}[{\cite[Lemma~2.2]{huang2026auction}}]\label{lem:gs}
Let $\M=(E,\cI)$ be a matroid and let $S\subseteq E$.
Let $p,p'\in\R^E$ be such that $p(e)\le p'(e)$ for all $e\in E$ and $p(e)=p'(e)$ for all $e\in S$.
If $S$ is contained in a $p$-minimum base of $\M$, then it is also contained in a $p'$-minimum base of $\M$.
\end{lemma}

Applying Lemma~\ref{lem:gs} to $p=-q$ and $p'=-q'$, we obtain the following lemma for maximum bases.

\begin{lemma}\label{lem:gs2}
Let $\M=(E,\cI)$ be a matroid and let $S\subseteq E$.
Let $q,q'\in\R^E$ be such that $q'(e)\le q(e)$ for all $e\in E$ and $q'(e)=q(e)$ for all $e\in S$.
If $S$ is contained in a $q$-maximum base of $\M$, then it is also contained in a $q'$-maximum base of $\M$.
\end{lemma}

\begin{lemma}[{\cite[Lemma~2.3]{huang2026auction}}]\label{lem:local}
Let $\M=(E,\cI)$ be a matroid and let $w\in\R^E$.
A set $S\in\cI$ is a $w$-maximum independent set of $\M$ if and only if
(i) $w(e)\le 0$ for any $e\in E\setminus S$ with $S+e\in\cI$,
(ii) $w(e)\ge 0$ for any $e\in S$, and
(iii) $w(e)\le w(f)$ for any $e\in E\setminus S$ and $f\in S$ with $S+e-f\in\cI$.
\end{lemma}

\begin{lemma}[Binary search technique~{\cite[Lemma~2.4\,(1)]{huang2026auction}}; from {\cite{nguyen2019note, chakrabarty2019faster, blikstad2021breaking_STOC}}]\label{lem:bs}
Let $\M=(E,\cI)$ be a matroid, let $w\in\R^E$, let $S\in\cI$, $R\subseteq S$, and $v\in E\setminus S$.
Using $O(\log|R|)$ independence queries, we can find an element $u\in R$ minimizing $w(u)$ subject to $S+v-u\in\cI$ (if such an element exists).
\end{lemma}

\begin{lemma} \label{lem:test}
Let $\M=(E,\cI)$ be a matroid with base family $\cB$, let $q\in\R^E$, let $Y\subseteq E$ be contained in a $q$-maximum base $Z$ of $\M$, and let $e\in E\setminus Y$.
Then $Y+e$ is contained in a $q$-maximum base of $\M$ if and only if there exists $e'\in Z\setminus Y$ such that $Z-e'+e\in\cB$ and $q(e')=q(e)$ (possibly $e'=e$).
\end{lemma}

\begin{proof}
We first show the ``if'' part.
Suppose that there exists $e'\in Z\setminus Y$ such that $Z-e'+e\in\cB$ and $q(e')=q(e)$.
Then $q(Z-e'+e)=q(Z)-q(e')+q(e)=q(Z)$, and hence $Z-e'+e$ is a $q$-maximum base of $\M$.
Since $e'\notin Y$, the base $Z-e'+e$ contains $Y+e$.

We next show the ``only if'' part.
Suppose that $Y+e$ is contained in a $q$-maximum base $Z'$ of $\M$.
If $e\in Z$, then $e'=e$ satisfies the required conditions, because $e\in Z\setminus Y$ and $Z-e+e=Z\in\cB$.
Suppose that $e\notin Z$.
Then $e\in Z'\setminus Z$.
By the simultaneous exchange property applied to the bases $Z'$ and $Z$ and the element $e\in Z'\setminus Z$, there exists $e'\in Z\setminus Z'$ such that both $Z'-e+e'$ and $Z+e-e'$ are bases of $\M$.
Since $Z$ and $Z'$ are $q$-maximum bases, we have $q(Z+e-e')\le q(Z)$ and $q(Z'-e+e')\le q(Z')$, that is, $q(e)\le q(e')$ and $q(e')\le q(e)$.
Hence $q(e')=q(e)$.
Moreover, $e'\in Z\setminus Z'\subseteq Z\setminus Y$ because $Y\subseteq Z'$.
Therefore, $e'$ satisfies the required conditions.
\end{proof}

% ====================================================================
\section{The Algorithm}\label{sec:algo}

\subsection{Prices and Utilities}

In the algorithm, we maintain a common independent set $X\in\cI_1\cap\cI_2$ and a nonnegative value $\pi(e)\in\R_+$ for each $e\in X$.
For notational convenience, we write $\pi(e)=\bot$ for $e\in E\setminus X$.
For $X$ and $\pi$, we define the vector $\bp\in\R_+^E$ by
\begin{equation}\label{eq:barprice}
  \bp(e):=
  \begin{cases}
    \pi(e) & \text{if } e\in X,\\
    0 & \text{if } e\in E\setminus X \text{ and } X+e\in\cI_2,\\
    \min\{\pi(f)\mid f\in C_2(X\mid e)-e\} & \text{if } e\in E\setminus X \text{ and } X+e\notin\cI_2.
  \end{cases}
\end{equation}
Note that, for $e \in E \setminus X$ with $X + e \notin \cI_2$, there always exists an element in $C_2(X\mid e)-e$, because we have assumed that $\M_2$ has no loops.
We define the \emph{price vector $p\in\R_+^E$ induced by $X$ and $\pi$} by adding $\eps w(e)$ to $\bp(e)$ for each $e\in E\setminus X$.
That is,
\begin{equation}\label{eq:price}
  p(e)=
  \begin{cases}
    \pi(e) & \text{if } e\in X,\\
    \eps w(e) & \text{if } e\in E\setminus X \text{ and } X+e\in\cI_2,\\
    \min\{\pi(f)\mid f\in C_2(X\mid e)-e\}+\eps w(e) & \text{if } e\in E\setminus X \text{ and } X+e\notin\cI_2.
  \end{cases}
\end{equation}
We omit ``induced by $X$ and $\pi$'' if no confusion arises.

These vectors have the following auction interpretation.
There are two players, Alice and Bob, and the weight $w(e)$ is interpreted as Alice's valuation of the item $e\in E$.
The set $X$ consists of the items that Alice and Bob have agreed to trade so far, and $\pi(e)$ is the agreed price of $e\in X$.
To sell an additional item $e\in E\setminus X$ with $X+e\notin\cI_2$, Bob has to remove some item $f\in X$ with $X+e-f\in\cI_2$.
The value $\bp(e)$ is the minimum price of such an item, and $\bp(e)=0$ if $X+e\in\cI_2$.
Bob sells $e$ at the price $p(e)$, which adds $\eps w(e)$ to $\bp(e)$.
The added term $\eps w(e)$ plays the same role as the price increase that Algorithm~\cite[Algorithm~2]{huang2026auction} applies when Bob removes an item from $X$.

For $u>0$, we define $\rup{u}:=(1+\eps)^{\lceil\log_{1+\eps}u\rceil}$, which rounds $u$ up to the nearest power of $1+\eps$, and $\rup{u}:=0$ for $u\le0$.
Let $\Lambda:=\{0\}\cup\{(1+\eps)^k\mid k\in\Z\}$ denote the set of all possible values of $\rup{\cdot}$.
That is, $\rup{u}=\min\{\lambda\in\Lambda\mid\lambda\ge u\}$ for every $u\in\R$.

We use two basic properties of the rounding.
\begin{enumerate}[label=($\Lambda$\arabic*),leftmargin=*]
\item\label{grid:round} $\max\{u,0\}\le\rup{u}\le(1+\eps)\cdot\max\{u,0\}$.
\item\label{grid:gap} If $\lambda,\lambda'\in\Lambda$ and $\lambda<\lambda'$, then $\lambda\le\lambda'/(1+\eps)$.
\end{enumerate}

We define the \emph{utility vector $q\in\Lambda^E$ induced by $X$ and $\pi$} as follows:
\begin{equation}\label{eq:utility}
  q(e):=\rup{w(e)-p(e)}\qquad(e\in E).
\end{equation}

In the auction interpretation, $q(e)$ is the rounded value of the utility $w(e)-p(e)$ that Alice gains by buying $e$ at the price $p(e)$.

The following lemma captures the effect of the rounding.
Roughly speaking, the rounding ensures that a utility can decrease only by a factor of at least $1+\eps$, and the lemma makes this guarantee usable in the analysis.
We apply it to $\M_1$ in the analysis of our algorithm (Section~\ref{sec:round}).

\begin{lemma}\label{lem:exchange}
Let $\M=(E,\cI)$ be a matroid with base family $\cB$, let $q\in\Lambda^E$ and $q'\in\R^E$, and let $Y\subseteq E$ be contained in a $q$-maximum base of $\M$.
Suppose that
\begin{enumerate}[label=(H\arabic*),leftmargin=*]
\item\label{hyp:mono} $q'(e)\le q(e)$ for every $e\in E$, and
\item\label{hyp:drop} $q'(e)\le q(e)/(1+\eps)$ for every $e\in E\setminus Y$ such that $Y+e$ is contained in a $q$-maximum base.
\end{enumerate}
Then
\[
  \max_{B\in\cB}q'(B)-q(Y)\;\le\;\frac{1}{1+\eps}\Bigl(\max_{B\in\cB}q(B)-q(Y)\Bigr).
\]
\end{lemma}
\begin{proof}
Let $C$ be a $q$-maximum base of $\M$ containing $Y$, so that $q(C)=\max_{B\in\cB}q(B)$.
Define $\hat q\in\R^E$ by
\[
  \hat q(e):=
  \begin{cases}
    q(e) & \text{if } e\in Y,\\
    q(e)/(1+\eps) & \text{if } e\in C\setminus Y,\\
    q'(e) & \text{if } e\in E\setminus C.
  \end{cases}
\]
We first show that $q'(e)\le \hat q(e)$ for every $e\in E$.
For $e\in Y$, this is~\ref{hyp:mono}.
For $e\in C\setminus Y$, the $q$-maximum base $C$ contains $Y+e$, and hence $q'(e)\le q(e)/(1+\eps)=\hat q(e)$ by~\ref{hyp:drop}.
For $e\in E\setminus C$, we have $\hat q(e)=q'(e)$ by the definition of $\hat q$.
Moreover, $\hat q(C)=q(Y)+\bigl(q(C)-q(Y)\bigr)/(1+\eps)$ by the definition of $\hat q$.
Hence, if $C$ is a $\hat q$-maximum independent set of $\M$, then
\[
  \max_{B\in\cB}q'(B)\;\le\;\max_{B\in\cB}\hat q(B)\;\le\;\hat q(C)\;=\;q(Y)+\frac{q(C)-q(Y)}{1+\eps},
\]
which is the assertion.
It remains to show that $C$ is a $\hat q$-maximum independent set of $\M$.

To this end, we check the three conditions of Lemma~\ref{lem:local} with $w=\hat q$ and $S=C$.
Condition~(i) is vacuous because $C$ is a base.
Condition~(ii) holds because $\hat q(c)\ge0$ for every $c\in C$, which follows from $q(c)\ge0$.
For condition~(iii), let $e\in E\setminus C$ and $c\in C$ with $C-c+e\in\cI$.
Then $C-c+e$ is a base of $\M$, and hence $q(e)\le q(c)$ because $C$ is a $q$-maximum base.
We also have $\hat q(e)=q'(e)\le q(e)$ by~\ref{hyp:mono}.
We show that $\hat q(e)\le \hat q(c)$ by distinguishing three cases.
\begin{itemize}[leftmargin=*]
\item If $c\in Y$, then $\hat q(c)=q(c)\ge q(e)\ge \hat q(e)$.
\item If $c\notin Y$ and $q(e)<q(c)$, then $q(e)\le q(c)/(1+\eps)$ by~\ref{grid:gap} because $q(e),q(c)\in\Lambda$, and hence $\hat q(e)\le q(e)\le q(c)/(1+\eps)=\hat q(c)$.
\item If $c\notin Y$ and $q(e)=q(c)$, then $C-c+e$ is a $q$-maximum base of $\M$ containing $Y+e$, and hence $\hat q(e)=q'(e)\le q(e)/(1+\eps)=q(c)/(1+\eps)=\hat q(c)$ by~\ref{hyp:drop}.
\end{itemize}
Therefore, $C$ is a $\hat q$-maximum independent set of $\M$ by Lemma~\ref{lem:local}, which completes the proof.
\end{proof}

\subsection{Algorithm Description}

Our algorithm is presented in Algorithm~\ref{alg:weighted}.
We call each iteration of the outer for-loop a \emph{round}, and we denote by $X^t$ the set $X$ at the end of round $t$, where $X^0=\emptyset$.
In each round, $X^*$, $\pi^*$, $p^*$, and $q^*$ denote $X$, $\pi$, $p$, and $q$ at the beginning of the round. 
They stay fixed during the round, whereas $X$, $\pi$, $p$, and $q$ change.
The number of rounds is $T:=\bigl\lceil 10\,(\eps^{-1}\log (r_1(E))+\eps^{-2})\bigr\rceil$.

\begin{algorithm}[t]
\caption{Weighted Matroid Intersection with Nearly Linear Independence Queries}\label{alg:weighted}
\KwInput{Two matroids $\M_1=(E,\cI_1)$ and $\M_2=(E,\cI_2)$, and weight $w(e)\in \R_+$ for $e\in E$}
\KwOutput{A common independent set $X\in\cI_1\cap\cI_2$}
Let $X\gets\emptyset$\;
\For{$t\gets 1$ \KwTo $T$}{
  $X^*,Y\gets X$\;
  $\pi^*(e)\gets\pi(e)$ for each $e\in E$\;
  \tcp{We keep explicitly $\pi^*(e)$ only for $e\in X^*$, as $\pi^*(e)=\bot$ for $e\in E\setminus X^*$}
  \tcp{Let $q^*$ denote the utility vector induced by $X^*$ and $\pi^*$}
  \ForEach{$e\in E\setminus X^*$}{
    \If{$q(e)=q^*(e)>0$ and $Y+e$ is contained in a $q^*$-maximum base of $\M_1$\label{line:cond}}{
      $Y\gets Y+e$\;\label{line:buy}
      $\pi(e)\gets p(e)$\;\label{line:pi}
      \eIf{$X+e\in\cI_2$}{
        $X\gets X+e$\;\label{line:add}
      }{
        Find $f\in C_2(X\mid e)-e$ minimizing $\pi(f)$\;
        $X\gets X+e-f$\;\label{line:exchange}
        $\pi(f)\gets\bot$\;\label{line:buyend}
      }
      \tcp{The price vector $p$ and the utility vector $q$ induced by $X$ and $\pi$ are updated implicitly}
    }
  }
  $X^t\gets X$\;
}
\Return{$X^t$ maximizing $w(X^t)$ over $t\in\{0,1,\dots,T\}$}\;
\end{algorithm}

In a round, the algorithm scans each $e\in E\setminus X^*$ once, in an arbitrary order.
If $q(e)=q^*(e)>0$, that is, the utility of $e$ is positive and has not changed since the beginning of the round, and $Y+e$ is contained in a $q^*$-maximum base of $\M_1$ (Line~\ref{line:cond}), then $e$ is added to $Y$ and $\pi(e)$ is set to the current price $p(e)$ (Lines~\ref{line:buy} and~\ref{line:pi}).
Then, $e$ is added to $X$.
If $X+e\notin\cI_2$, an element $f\in C_2(X\mid e)-e$ minimizing $\pi(f)$ is removed from $X$ at the same time (Lines~\ref{line:add}--\ref{line:buyend}).

In the auction interpretation, a round proceeds as follows.
Alice aims to buy a $q^*$-maximum base of $\M_1$.
She buys an item $e\in E\setminus X^*$ if its utility is positive, if the utility has not changed since the beginning of the round, and if $e$ and the items bought so far lie in a common $q^*$-maximum base of $\M_1$.
The last condition says that buying $e$ keeps this aim achievable.
When Alice buys $e$ and $X+e\notin\cI_2$, Bob removes a cheapest item $f\in X$ with $X+e-f\in\cI_2$.

\section{Analysis}\label{sec:analysis}

% Throughout this section, let
% \[
%   U:=\pi(X)+\max_{B\in\cB_1}q(B).
% \]

Our analysis mainly uses the following quantity.
For $X$ and $\pi$, let
\[
  U(X,\pi):=\pi(X)+\max_{B\in\cB_1}q(B),
\]
where $q$ is the utility vector induced by $X$ and $\pi$.
We simply write $U:=U(X,\pi)$ if no confusion arises, where $X$ and $\pi$ are the current ones maintained by Algorithm~\ref{alg:weighted}.

\subsection{Useful Properties}\label{sec:basic}

In this subsection, we show several properties that are maintained throughout the execution of Algorithm~\ref{alg:weighted}.

\begin{observation}\label{obs:basic}
Throughout the execution of Algorithm~\ref{alg:weighted}, the following hold.
\begin{enumerate}[label=(\alph*),leftmargin=*]
\item\label{obs:range} $X\in\cI_2$, and $\eps w(e)\le\pi(e)<w(e)$ for each $e\in X$. In particular, $\pi(X)\le w(X)$.
\item\label{obs:charge} When Lines~\ref{line:buy}--\ref{line:buyend} are executed for $e\in E\setminus X^*$, the value $p(e)$ is unchanged and $\pi(X)$ increases by exactly $\eps w(e)$.
\end{enumerate}
\end{observation}
\begin{proof}
When Lines~\ref{line:buy}--\ref{line:buyend} are executed for $e$, we have $e\notin X$.
Indeed, $e\notin X^*$, and in a round an element is added to $X$ only when these lines are executed for it, which happens at most once because each element of $E\setminus X^*$ is scanned once.

(a)~Lines~\ref{line:add} and~\ref{line:exchange} preserve $X\in\cI_2$ by Observation~\ref{obs:circuit}.
For $e\in X$, the value $\pi(e)$ was assigned in Line~\ref{line:pi} and has not changed since then.
At the time of the assignment, we have $e\notin X$, so $\pi(e)=p(e)=\bp(e)+\eps w(e)\ge\eps w(e)$, where $\bp(e)\ge0$ because $\bp(e)$ is either $0$ or a value of $\pi$ on $X$, and $\pi\ge0$ on $X$ by induction.
Moreover, $q(e)>0$ in Line~\ref{line:cond} means $w(e)-p(e)>0$ by~\ref{grid:round}, so $\pi(e)=p(e)<w(e)$.

(b)~In Line~\ref{line:pi}, $\pi(e)$ is set to $p(e)$.
After Line~\ref{line:add} or~\ref{line:exchange}, we have $e\in X$ and $p(e)=\pi(e)$.
Hence, $p(e)$ is unchanged by the execution of Lines~\ref{line:buy}--\ref{line:buyend}.
If $X+e\in\cI_2$, then $\bp(e)=0$, and $\pi(X)$ increases by $\pi(e)=\eps w(e)$.
Otherwise, $f$ minimizes $\pi$ on $C_2(X\mid e)-e$, so $\pi(f)=\bp(e)$, and $\pi(X)$ increases by $\pi(e)-\pi(f)=\eps w(e)$.
\end{proof}

We next show that prices never decrease and utilities never increase.
In the auction interpretation, Bob never lowers the price of an item.

\begin{lemma}\label{lem:monotone}
For each $e'\in E$, $p(e')$ is monotonically non-decreasing and $q(e')$ is monotonically non-increasing throughout the execution of Algorithm~\ref{alg:weighted}.
\end{lemma}
\begin{proof}
For $z>0$, let $X_{\ge z}:=\{x\in X\mid\pi(x)\ge z\}$.
We first show that
\begin{equation}\label{eq:level-span}
  \bp(e')\ge z\quad\Longleftrightarrow\quad e'\in\spn_{\M_2}(X_{\ge z})\qquad(e'\in E,\ z>0).
\end{equation}
For $e'\in X$, the left-hand side is $\pi(e')\ge z$ by the definition of $\bp$.
The right-hand side also holds if and only if $\pi(e')\ge z$, that is, $e'\in X_{\ge z}$, because $X$ is independent and hence no element of $X\setminus X_{\ge z}$ is contained in $\spn_{\M_2}(X_{\ge z})$.
For $e'\notin X$, the right-hand side says that $X_{\ge z}+e'$ contains a circuit.
The only circuit in $X+e'$, if any, is $C_2(X\mid e')$, so this happens if and only if $X+e'\notin\cI_2$ and $C_2(X\mid e')-e'\subseteq X_{\ge z}$, that is, if and only if $X+e'\notin\cI_2$ and every element of $C_2(X\mid e')-e'$ has price at least $z$.
By the definition of $\bp$, this is exactly $\bp(e')\ge z$.

We next show that $\bp(e')$ is monotonically non-decreasing for each $e'\in E$.
Since $X$ and $\pi$ change only in Lines~\ref{line:buy}--\ref{line:buyend}, it suffices to show that $\bp(e')$ does not decrease when these lines are executed for an element $e\in E\setminus X^*$.
Let $\alpha$ be the value of $\bp(e')$ before the execution.
If $\alpha=0$, then $\bp(e')\ge \alpha$ after the execution because $\bp(e')\ge0$.
If $\alpha>0$, then $e'\in\spn_{\M_2}(X_{\ge \alpha})$ before the execution by~\eqref{eq:level-span}, and hence $\bp(e')\ge \alpha$ after the execution as long as $e'$ is still contained in $\spn_{\M_2}(X_{\ge a})$.
Therefore, it suffices to show that $\spn_{\M_2}(X_{\ge z})$ does not shrink by the execution for any $z>0$.
Let $\bp_e$ be the value of $\bp(e)$ before the execution.
The execution sets $\pi(e)=\bp_e+\eps w(e)$ in Line~\ref{line:pi}, and, if $X+e\notin\cI_2$, removes from $X$ an element $f$ with $\pi(f)=\bp_e$ in Line~\ref{line:exchange}.
If $X+e\in\cI_2$ or $z>\bp_e$, then $X_{\ge z}$ can only gain the element $e$, and hence $\spn_{\M_2}(X_{\ge z})$ does not shrink.
If $X+e\notin\cI_2$ and $z\le\bp_e$, then $X_{\ge z}$ becomes $X_{\ge z}-f+e$.
In this case, $X_{\ge z}$ contains $C_2(X\mid e)-e$ before the execution, because every $f'\in C_2(X\mid e)-e$ satisfies $\pi(f')\ge\bp_e\ge z$ by the definition of $\bp_e$.
Hence, $e\in\spn_{\M_2}(X_{\ge z})$, and $f\in\spn_{\M_2}(X_{\ge z}+e-f)$ because $f\in C_2(X\mid e)\subseteq X_{\ge z}+e$.
Thus, $X_{\ge z}-f+e$ and $X_{\ge z}$ have the same span.

We now show that $p(e')$ is monotonically non-decreasing for each $e'\in E$.
Recall that $p(e')=\bp(e')+\eps w(e')$ for $e'\notin X$ and $p(e')=\bp(e')$ for $e'\in X$.
Consider again the execution of Lines~\ref{line:buy}--\ref{line:buyend} for $e$, and let $f$ be the element removed from $X$ in Line~\ref{line:exchange} if $X+e\notin\cI_2$.
If $e'\notin\{e,f\}$, then whether $e'\in X$ does not change, and hence $p(e')$ does not decrease because $\bp(e')$ does not decrease.
If $e'=e$, then $p(e')$ is unchanged by Observation~\ref{obs:basic}\ref{obs:charge}.
If $e'=f$, then $f$ is removed from $X$, and hence $p(f)=\bp(f)+\eps w(f)$ after the execution, which is at least the value $p(f)=\bp(f)$ before the execution because $\bp(f)$ does not decrease.
Finally, $q(e')=\rup{w(e')-p(e')}$ is monotonically non-increasing for each $e'\in E$, because $p(e')$ is monotonically non-decreasing.
\end{proof}

The next lemma establishes the invariants of the algorithm.
In the auction interpretation, it says that Alice's aim of buying a $q^*$-maximum base of $\M_1$ stays achievable throughout each round.

\begin{lemma}\label{lem:inv}
In each round, $Y$ is contained in a $q^*$-maximum base of $\M_1$ before and after each iteration of the inner for-loop.
Moreover, $X\in\cI_1\cap\cI_2$ throughout, and at the beginning of each round $X$ is contained in a $q$-maximum base of $\M_1$.
\end{lemma}
\begin{proof}
We show by induction on the rounds that $X$ is contained in a $q$-maximum base of $\M_1$ at the beginning of each round and that the following hold before and after each iteration of the inner for-loop in the round:
\begin{enumerate}[label=(\arabic*),leftmargin=*]
\item\label{inv:Y} $X\cup X^*\subseteq Y$, and $Y$ is contained in a $q^*$-maximum base of $\M_1$;
\item\label{inv:owner} $q(x)=q^*(x)$ for each $x\in X$.
\end{enumerate}
At the beginning of the first round, $X=\emptyset$ is contained in every base.

Consider a round at whose beginning $X$ is contained in a $q$-maximum base.
Then \ref{inv:Y} and~\ref{inv:owner} hold at the beginning, because $X=X^*=Y$ and $q=q^*$.
An iteration changes nothing unless Lines~\ref{line:buy}--\ref{line:buyend} are executed for the scanned element $e$.
We show that their execution preserves both invariants.
Invariant~\ref{inv:Y} is preserved, because $e$ is added to $Y$ in Line~\ref{line:buy} and $Y+e$ is contained in a $q^*$-maximum base of $\M_1$ by the condition in Line~\ref{line:cond}.
For invariant~\ref{inv:owner}, consider the set $X$ after the execution.
Every element $x\neq e$ of this set was already in $X$ and keeps its value $\pi(x)$, and hence its utility $q(x)=\rup{w(x)-\pi(x)}$.
% The utility $q(e)$ keeps the value it had just before the execution (Observation~\ref{obs:basic}\ref{obs:charge}), which is $q^*(e)$ by Line~\ref{line:cond}.
By Observation~\ref{obs:basic}\ref{obs:charge}, $p(e)$ is unchanged by the execution, and hence so is $q(e)=\rup{w(e)-p(e)}$.
Since $q(e)=q^*(e)$ held in Line~\ref{line:cond}, we have $q(e)=q^*(e)$ after the execution.

% At the end of the round, $X$ is contained in a $q^*$-maximum base by~\ref{inv:Y}, $q=q^*$ on $X$ by~\ref{inv:owner}, and $q\le q^*$ on $E$ by Lemma~\ref{lem:monotone}.
At the end of the round, $X$ is contained in a $q^*$-maximum base of $\M_1$ by~\ref{inv:Y}, $q(x)=q^*(x)$ for every $x\in X$ by~\ref{inv:owner}, and $q(e')\le q^*(e')$ for every $e'\in E$ by Lemma~\ref{lem:monotone}.
By Lemma~\ref{lem:gs2}, $X$ is contained in a $q$-maximum base of $\M_1$, which completes the induction.
Finally, $X\in\cI_1$ by~\ref{inv:Y} and $X\in\cI_2$ by Observation~\ref{obs:basic}\ref{obs:range}.
\end{proof}

\subsection{Weight-Splitting Argument}\label{sec:certificate}

In this subsection, we prove the following lemma.
It adapts the weight-splitting argument in~\cite[Section~4]{huang2026auction}, with $\bp$ and $q$ in the roles of $p'$ and $w-p$.
The proof shows that
\[
  (1-\eps)\cdot w(\OPT_w)\;\le\;\bp(\OPT_w)+q(\OPT_w)\;\le\;\pi(X)+\max_{B\in\cB_1}q(B)\;=\;U(X, \pi).
\]

\begin{lemma}\label{lem:upper}
At any point of Algorithm~\ref{alg:weighted}, we have $(1-\eps)\cdot w(\OPT_w)\le U(X, \pi)$.
\end{lemma}
\begin{proof}
We bound $\bp(\OPT_w)$ and $q(\OPT_w)$ separately, and then use that $\bp(e)+q(e)\ge(1-\eps)\,w(e)$ for every $e\in E$.

We first show that $\bp(\OPT_w)\le\pi(X)$.
Recall that $X\in\cI_2$ by Observation~\ref{obs:basic}\ref{obs:range}.
To see that $X$ is a $\bp$-maximum independent set of $\M_2$, we check the three conditions of Lemma~\ref{lem:local} with $w=\bp$ and $S=X$.
Condition~(i) holds, because $\bp(e)=0$ for every $e\notin X$ with $X+e\in\cI_2$.
Condition~(ii) holds, because $\bp(e)=\pi(e)\ge0$ for every $e\in X$.
For condition~(iii), let $e\notin X$ and $f\in X$ with $X+e-f\in\cI_2$.
If $X+e\in\cI_2$, then $\bp(e)=0\le\pi(f)$.
Otherwise, $f\in C_2(X\mid e)-e$ by Observation~\ref{obs:circuit}, and hence $\bp(e)\le\pi(f)$ by the definition of $\bp$.
In both cases, $\bp(e)\le\pi(f)=\bp(f)$.
Therefore, $X$ is a $\bp$-maximum independent set of $\M_2$.
Since $\OPT_w\in\cI_2$, we obtain $\bp(\OPT_w)\le\bp(X)=\pi(X)$.

We next show that $q(\OPT_w)\le\max_{B\in\cB_1}q(B)$.
This holds because $q(e)\ge0$ for every $e\in E$ and $\OPT_w\in\cI_1$ is contained in a base of $\M_1$.

Finally, we combine the two bounds.
For every $e\in E$, we have $p(e)\le\bp(e)+\eps w(e)$ by the definition of $\bp$, and $q(e)\ge w(e)-p(e)$ by~\ref{grid:round}.
Hence, $(1-\eps)\,w(e)\le\bp(e)+q(e)$ for every $e\in E$.
Summing over $e\in\OPT_w$ and using the two bounds, we obtain
\[
  (1-\eps)\cdot w(\OPT_w)\le\bp(\OPT_w)+q(\OPT_w)\le\pi(X)+\max_{B\in\cB_1}q(B)=U(X, \pi).\qedhere
\]
\end{proof}

\subsection{Progress in One Round}\label{sec:round}

Our analysis mainly tracks how $U(X,\pi)$ changes during the execution of Algorithm~\ref{alg:weighted}.
In this subsection, we bound the change of $U(X,\pi)$ in one round (Lemma~\ref{lem:round}).
We prove Lemma~\ref{lem:round} by applying Lemma~\ref{lem:exchange} to $\M_1$.
To apply it, we need to verify hypotheses~\ref{hyp:mono} and~\ref{hyp:drop}.
Hypothesis~\ref{hyp:mono} follows from Lemma~\ref{lem:monotone}, and the next lemma provides hypothesis~\ref{hyp:drop}.
 
\begin{lemma}\label{lem:drop}
Consider a round, and let $Y$ and $q'$ be the set $Y$ and the utility vector at the end of the round.
If $e\in E\setminus Y$ and $Y+e$ is contained in a $q^*$-maximum base of $\M_1$, then $q'(e)\le q^*(e)/(1+\eps)$.
\end{lemma}
\begin{proof}
If $q^*(e)=0$, then $q'(e)=0$ by Lemma~\ref{lem:monotone} and $q'(e)\ge0$.
Suppose that $q^*(e)>0$.
Since $X^*\subseteq Y$, we have $e\in E\setminus X^*$, and hence the condition in Line~\ref{line:cond} was checked for $e$ in the round.
Let $Y_0$ be the set $Y$ at the time of this check.
Since $e\notin Y$ at the end of the round, Lines~\ref{line:buy}--\ref{line:buyend} were not executed for $e$, that is, the condition in Line~\ref{line:cond} did not hold for $e$.
On the other hand, since the set $Y$ only grows during the round, we have $Y_0\subseteq Y$, and hence $Y_0+e$ is contained in the $q^*$-maximum base of $\M_1$ that contains $Y+e$.
Therefore, $q(e)=q^*(e)>0$ did not hold at the time of the check.
Since $q^*(e)>0$, this means that $q(e)\neq q^*(e)$, and hence $q(e)<q^*(e)$ by Lemma~\ref{lem:monotone}.
By~\ref{grid:gap}, we obtain $q(e)\le q^*(e)/(1+\eps)$ because $q(e),q^*(e)\in\Lambda$.
Since $q(e)$ does not increase afterwards by Lemma~\ref{lem:monotone}, we obtain $q'(e)\le q^*(e)/(1+\eps)$.
\end{proof}

By applying Lemma~\ref{lem:exchange} to $\M_1$, we now bound the change of $U(X,\pi)$ in one round.
 
\begin{lemma}\label{lem:round}
Consider a round. Let $X^*$ and $X'$ be the sets $X$ at its beginning and its end, let $U$ and $U'$ be the values of $U$ at its beginning and its end, and let $\Delta:=\pi(X')-\pi^*(X^*)$ be the increase of $\pi(X)$ in the round. Then
\[
  U'\;\le\;\frac{U}{1+\eps}+\eps\,w(X^*)+2\Delta.
\]
\end{lemma}
\begin{proof}
Let $q'$ and $Y$ be the utility vector and the set $Y$ at the end of the round, let $A:=Y\setminus X^*$ be the set of elements added to $Y$ in the round, and let $P:=\pi^*(X^*)$.
Thus $U=P+\max_{B\in\cB_1}q^*(B)$ and $U'=P+\Delta+\max_{B\in\cB_1}q'(B)$.
Since Lines~\ref{line:buy}--\ref{line:buyend} are executed exactly once for each element of $A$, we have $\Delta=\eps w(A)$ by Observation~\ref{obs:basic}\ref{obs:charge}.
 
We first apply Lemma~\ref{lem:exchange} to $\M_1$, $q^*$, $q'$, and $Y$.
Its assumptions are satisfied as follows.
We have $q^*\in\Lambda^E$ by the definition of $q^*$.
The set $Y$ is contained in a $q^*$-maximum base of $\M_1$ by Lemma~\ref{lem:inv}.
Hypothesis~\ref{hyp:mono} holds by Lemma~\ref{lem:monotone}, and hypothesis~\ref{hyp:drop} holds by Lemma~\ref{lem:drop}.
Hence, Lemma~\ref{lem:exchange} gives
\[
  \max_{B\in\cB_1}q'(B)\;\le\;q^*(Y)+\frac{U-P-q^*(Y)}{1+\eps}.
\]
Adding $P+\Delta$ to both sides, we obtain
\begin{equation}\label{eq:round-mid}
  U'\;\le\;\frac{U}{1+\eps}+\Delta+\frac{\eps}{1+\eps}\bigl(P+q^*(Y)\bigr).
\end{equation}
 
We next show that
\begin{equation}\label{eq:Y-cheap}
  P+q^*(Y)\;\le\;(1+\eps)\,w(Y).
\end{equation}
By~\ref{grid:round}, $q^*(y)\le(1+\eps)\cdot\max\{w(y)-p^*(y),0\}$ for every $y\in E$.
For $y\in X^*$, we have $p^*(y)=\pi^*(y)<w(y)$ by Observation~\ref{obs:basic}\ref{obs:range}, and hence $q^*(y)\le(1+\eps)\bigl(w(y)-\pi^*(y)\bigr)$.
Adding $\pi^*(y)$ to both sides, we obtain $\pi^*(y)+q^*(y)\le(1+\eps)\,w(y)-\eps\pi^*(y)\le(1+\eps)\,w(y)$, where the last inequality follows from $\pi^*(y)\ge0$.
For $y\in A$, we have $q^*(y)\le(1+\eps)\,w(y)$ because $p^*(y)\ge0$.
Summing these inequalities over $Y=X^*\cup A$ gives~\eqref{eq:Y-cheap}, because $P=\pi^*(X^*)$.
 
Finally, by~\eqref{eq:Y-cheap} and $\Delta=\eps w(A)$, we have
\[
  \frac{\eps}{1+\eps}\bigl(P+q^*(Y)\bigr)\;\le\;\eps\,w(Y)\;=\;\eps\,w(X^*)+\eps\,w(A)\;=\;\eps\,w(X^*)+\Delta.
\]
Substituting this into~\eqref{eq:round-mid}, we obtain
\[
  U'\;\le\;\frac{U}{1+\eps}+\eps\,w(X^*)+2\Delta,
\]
which completes the proof.
\end{proof}

\subsection{Approximation Ratio}\label{sec:ratio}

In this subsection, we show that Algorithm~\ref{alg:weighted} with $T:=\bigl\lceil 10\,(\eps^{-1}\log (r_1(E))+\eps^{-2})\bigr\rceil$ rounds returns a $(1-3\eps)$-approximate solution.

We first explain a rough idea of the proof.
Suppose that a round starts with $w(X)<(1-3\eps)\,w(\OPT_w)$.
Let $g$ and $g'$ be the values of $U-(1-2\eps)\,w(\OPT_w)$, and let $P$ and $P'$ be the values of $\pi(X)$, at the beginning and the end of the round.
Then Lemma~\ref{lem:round} yields
\[
  g'\;\le\;\frac{g}{1+\eps}+2(P'-P).
\]
On the other hand, by Lemma~\ref{lem:upper}, we have $U\ge(1-\eps)\,w(\OPT_w)$ at any point, and hence
\[
  U-(1-2\eps)\,w(\OPT_w)\;\ge\;(1-\eps)\,w(\OPT_w)-(1-2\eps)\,w(\OPT_w)\;=\;\eps\,w(\OPT_w).
\]
In particular, $g\ge\eps\,w(\OPT_w)$, and hence $2(P'-P)\le\frac{2(P'-P)}{\eps\,w(\OPT_w)}\,g$.
Substituting this into the inequality $g'\le\frac{g}{1+\eps}+2(P'-P)$ above, we obtain
\[
  g'\;\le\;\biggl(\frac{1}{1+\eps}+\frac{2(P'-P)}{\eps\,w(\OPT_w)}\biggr)\,g\;\le\;\frac{g}{1+\eps}\,e^{O((P'-P)/(\eps\,w(\OPT_w)))}.
\]
Now suppose that the first $T$ rounds all start with $w(X)<(1-3\eps)\,w(\OPT_w)$.
For $t\in\{0,1,\dots,T\}$, let $g_t$ and $P^t$ be the values of $U-(1-2\eps)\,w(\OPT_w)$ and $\pi(X)$ at the end of round $t$, where the end of round $0$ refers to the beginning of the algorithm.
Applying the above inequality to rounds $t=T,T-1,\dots,1$, we obtain
\[
  g_T\;\le\;g_0\cdot(1+\eps)^{-T}\exp\Biggl(O\Biggl(\frac{1}{\eps\,w(\OPT_w)}\sum_{t=1}^{T}(P^t-P^{t-1})\Biggr)\Biggr).
\]
Since $P^0=\pi(\emptyset)=0$, the sum equals $P^T$.
Moreover, $P^T\le w(X^T)\le w(\OPT_w)$, and $g_T\ge\eps\,w(\OPT_w)$ as shown above.
Hence
\[
  \eps\,w(\OPT_w)\;\le\;g_0\cdot(1+\eps)^{-T}\cdot e^{O(1/\eps)}.
\]
Since $g_0<r_1(E)\,w(\OPT_w)$, the right-hand side is less than $\eps\,w(\OPT_w)$ for our choice of $T$, a contradiction.

We now formalize the above argument.
Instead of deriving a contradiction, the proof below argues directly on the first $t$ with $w(X^t)\ge(1-3\eps)\,w(\OPT_w)$.
 
\begin{proposition} \label{prop:approx}
Let $T:=\bigl\lceil 10\,(\eps^{-1}\log (r_1(E))+\eps^{-2})\bigr\rceil$.
Then, Algorithm~\ref{alg:weighted} returns a common independent set $X$ with $w(X)\ge(1-3\eps)\,w(\OPT_w)$.
\end{proposition}

\begin{proof}
The output is a common independent set by Lemma~\ref{lem:inv}.
Let $\optw:=w(\OPT_w)$.
We assume that $\eps<1/3$ without loss of generality.
For $t\in\{0,1,\dots,T\}$, let $P^t$ and $U^t$ be the values of $\pi(X)$ and $U$ at the end of round $t$, where the end of round $0$ refers to the beginning of the algorithm.
Let
\[
  g_t:=U^t-(1-2\eps)\cdot\optw\qquad(t\in\{0,1,\dots,T\}).
\]
By Lemma~\ref{lem:upper}, we have $U^t\ge(1-\eps)\cdot\optw$ for every $t\in\{0,1,\dots,T\}$.
Hence, for every $t\in\{0,1,\dots,T\}$,
\[
  g_t\;\ge\;(1-\eps)\cdot\optw-(1-2\eps)\cdot\optw\;=\;\eps \cdot  \optw.
\]
 
Let
\[
  m:=\min\{\,t\in\{0,1,\dots,T\}\mid w(X^t)\ge(1-3\eps)\cdot\optw\,\},
\]
where we set $m:=T$ if no such $t$ exists.
In both cases, $w(X^t)<(1-3\eps)\cdot\optw$ holds for every $t\in\{0,1,\dots,m-1\}$.
In the rest of the proof, we show that $m<T$.
This suffices to prove the proposition.
Indeed, if $m<T$, then $w(X^m)\ge(1-3\eps)\cdot\optw$ by the definition of $m$, and the algorithm returns $X^t$ maximizing $w(X^t)$ over $t\in\{0,1,\dots,T\}$.
 
We first show that
\begin{equation}\label{eq:gap-rec}
  g_t\;\le\;\frac{g_{t-1}}{1+\eps}+2(P^t-P^{t-1})\qquad(t\in\{1,\dots,m\}).
\end{equation}
Fix $t\in\{1,\dots,m\}$, and apply Lemma~\ref{lem:round} to round $t$, in which $X^*=X^{t-1}$, $U=U^{t-1}$, $U'=U^t$, and $\Delta=P^t-P^{t-1}$.
By the definition of $m$, we have $w(X^{t-1})<(1-3\eps)\cdot\optw$.
Hence Lemma~\ref{lem:round} gives
\[
  U^t\;\le\;\frac{U^{t-1}}{1+\eps}+\eps\,w(X^{t-1})+2(P^t-P^{t-1})\;\le\;\frac{U^{t-1}}{1+\eps}+\eps\,(1-3\eps)\cdot\optw+2(P^t-P^{t-1}).
\]
Moreover, since $(1+\eps)(1-3\eps)=1-2\eps-3\eps^2\le1-2\eps$, we have
\[
  \eps\,(1-3\eps)\cdot\optw\;\le\;\frac{\eps}{1+\eps}\,(1-2\eps)\cdot\optw\;=\;(1-2\eps)\cdot\optw-\frac{(1-2\eps)\cdot\optw}{1+\eps}.
\]
Combining these two inequalities and subtracting $(1-2\eps)\cdot\optw$ from both sides, we obtain
\[
  g_t\;\le\;\frac{U^{t-1}-(1-2\eps)\cdot\optw}{1+\eps}+2(P^t-P^{t-1})\;=\;\frac{g_{t-1}}{1+\eps}+2(P^t-P^{t-1}),
\]
which is~\eqref{eq:gap-rec}.
 
We next show that
\begin{equation}\label{eq:gap-prod}
  g_m\;\le\;g_0\cdot(1+\eps)^{-m}\exp\biggl(\frac{2(1+\eps)}{\eps\cdot\optw}\,P^m\biggr).
\end{equation}
Fix $t\in\{1,\dots,m\}$.
Recall that $g_{t-1}\ge\eps \cdot\optw$, which we derived above from Lemma~\ref{lem:upper}.
We also have $P^t-P^{t-1}\ge0$, because $\pi(X)$ changes only when Lines~\ref{line:buy}--\ref{line:buyend} are executed, and each execution increases $\pi(X)$ by Observation~\ref{obs:basic}\ref{obs:charge}.
Hence,
\[
  2(P^t-P^{t-1})\;=\;\frac{g_{t-1}}{1+\eps}\cdot\frac{2(1+\eps)}{g_{t-1}}\,(P^t-P^{t-1})\;\le\;\frac{g_{t-1}}{1+\eps}\cdot\frac{2(1+\eps)}{\eps \cdot\optw}\,(P^t-P^{t-1}).
\]
Substituting this into~\eqref{eq:gap-rec} and using $1+x\le e^x$, we obtain
\[
  g_t\;\le\;\frac{g_{t-1}}{1+\eps}\biggl(1+\frac{2(1+\eps)}{\eps \cdot\optw}\,(P^t-P^{t-1})\biggr)\;\le\;\frac{g_{t-1}}{1+\eps}\exp\biggl(\frac{2(1+\eps)}{\eps \cdot\optw}\,(P^t-P^{t-1})\biggr).
\]
Applying this inequality repeatedly for $t=m,m-1,\dots,1$, we obtain
\[
  g_m\;\le\;g_0\cdot(1+\eps)^{-m}\exp\Biggl(\frac{2(1+\eps)}{\eps \cdot\optw}\sum_{t=1}^{m}(P^t-P^{t-1})\Biggr).
\]
Since $\sum_{t=1}^{m}(P^t-P^{t-1})=P^m-P^0$ and $P^0=\pi(\emptyset)=0$, this proves~\eqref{eq:gap-prod}.

We next show that $P^m\le \optw$.
By Observation~\ref{obs:basic}\ref{obs:range}, we have $P^m=\pi(X^m)\le w(X^m)$.
Moreover, $X^m\in\cI_1\cap\cI_2$ by Lemma~\ref{lem:inv}, and $\optw$ is the maximum weight of a common independent set.
Hence, $w(X^m)\le \optw$, and therefore $P^m\le \optw$.
Recall that $g_m\ge\eps \cdot\optw$, which we derived above from Lemma~\ref{lem:upper}.
Substituting this and $P^m\le \optw$ into~\eqref{eq:gap-prod}, we obtain
\begin{equation}\label{eq:gap-final}
  \eps \cdot\optw\;\le\;g_m\;\le\;g_0\cdot(1+\eps)^{-m}\,e^{2(1+\eps)/\eps}.
\end{equation}
 
We next show that $g_0<r_1(E)\,W$.
Since $(1-2\eps)\cdot\optw>0$ by $\eps<1/3$, we have $g_0=U^0-(1-2\eps)\cdot\optw\le U^0$.
At the beginning of the algorithm, we have $X=\emptyset$, and hence $p(e)=\eps w(e)$ for every $e\in E$.
Therefore, every $e\in E$ satisfies
\[
  q(e)\;=\;\rup{(1-\eps)\,w(e)}\;\le\;(1+\eps)(1-\eps)\,w(e)\;<\;w(e)\;\le\;\optw.
\]
Here, the first inequality follows from~\ref{grid:round} and the last from $w(e)\le w(\OPT_w)$.
Since $\pi(\emptyset)=0$ and every base of $\M_1$ consists of $r_1(E)$ elements, we obtain
\[
  g_0\;\le\;U^0\;=\;\max_{B\in\cB_1}q(B)\;<\;r_1(E)\cdot\optw.
\]
 
Finally, substituting $g_0<r_1(E)\cdot\optw$ into~\eqref{eq:gap-final}, we obtain
\[
  \eps \cdot\optw\;<\;r_1(E)\cdot\optw\cdot(1+\eps)^{-m}\,e^{2(1+\eps)/\eps}.
\]
Dividing both sides by $\optw$ and multiplying both sides by $(1+\eps)^m/\eps$, we obtain
\[
  (1+\eps)^m\;<\;\frac{r_1(E)}{\eps}\,e^{2(1+\eps)/\eps}.
\]
Taking logarithms of both sides, we obtain
\[
  m\log(1+\eps)\;<\;\log (r_1(E))+\log\frac{1}{\eps}+\frac{2(1+\eps)}{\eps}.
\]
Since $0<\eps<1/3$, we have $\log(1+\eps)\ge\eps/2$, $\log(1/\eps)<1/\eps$, and $2(1+\eps)/\eps<3/\eps$.
Hence
\[
  m\;<\;\frac{2}{\eps}\Bigl(\log (r_1(E))+\frac{1}{\eps}+\frac{3}{\eps}\Bigr)\;=\;2\eps^{-1}\log (r_1(E))+8\eps^{-2}\;\le\;10\,(\eps^{-1}\log (r_1(E))+\eps^{-2})\;\le\;T,
\]
where the last inequality follows from the definition $T=\bigl\lceil 10\,(\eps^{-1}\log (r_1(E))+\eps^{-2})\bigr\rceil$.
Therefore, $m<T$, which completes the proof.
\end{proof}

\section{Query Complexity}\label{sec:query}

In this section, we show how to implement each round of Algorithm~\ref{alg:weighted} with $O(n\log r)$ independence queries, and prove Theorem~\ref{thm:main}.
Throughout a round, we maintain a $q^*$-maximum base $Z$ of $\M_1$ with $Y\subseteq Z$.
Using $Z$ and Lemma~\ref{lem:test}, we check the condition in Line~\ref{line:cond} with $O(\log r)$ independence queries by the binary search technique (Lemma~\ref{lem:bs}).

\begin{lemma}\label{lem:query}
Each round of Algorithm~\ref{alg:weighted} can be implemented using $O(n\log r)$ independence queries.
\end{lemma}

\begin{proof}
We first show that, given $X$ and $\pi$, the price $p(e)$ and the utility $q(e)$ of an element $e\in E$ can be computed with $O(\log r)$ independence queries.
For $e\in X$, we have $p(e)=\pi(e)$, which is stored.
For $e\notin X$, one query decides whether $X+e\in\cI_2$.
If $X+e\in\cI_2$, then $p(e)=\eps w(e)$.
Otherwise, by Observation~\ref{obs:circuit}, an element $f\in C_2(X\mid e)-e$ minimizing $\pi(f)$ is an element $f\in X$ minimizing $\pi(f)$ subject to $X+e-f\in\cI_2$, and hence it can be found by Lemma~\ref{lem:bs} with $S=R=X$, $v=e$, and $w=\pi$ using $O(\log|X|)=O(\log r)$ queries, where $|X|\le r$ because $X\in\cI_1\cap\cI_2$.
Then, $p(e)=\pi(f)+\eps w(e)$, and $q(e)$ is obtained from $p(e)$ by~\eqref{eq:utility}.
In the same way, $q^*(e)$ can be computed from $X^*$ and $\pi^*$ with $O(\log r)$ queries.

We next describe the initialization of a round.
We compute $q^*(e)$ for all $e\in E$ with $O(n\log r)$ queries.
We then compute a base $Z$ of $\M_1$ as follows.
Starting from $Z=X^*$, we scan the elements of $E\setminus X^*$ in non-increasing order of $q^*$, and add each scanned element $e$ to $Z$ if $Z+e\in\cI_1$.
This requires $n$ queries.
This procedure is the greedy algorithm for the contraction $\M_1/X^*$, and hence $Z\setminus X^*$ is a $q^*$-maximum base of $\M_1/X^*$.
Let $B$ be any base of $\M_1$ containing $X^*$.
Then $B\setminus X^*$ is a base of $\M_1/X^*$, and hence $q^*(B)=q^*(X^*)+q^*(B\setminus X^*)\le q^*(X^*)+q^*(Z\setminus X^*)=q^*(Z)$.
That is, $Z$ maximizes $q^*$ among the bases of $\M_1$ containing $X^*$.
By Lemma~\ref{lem:inv}, there exists a $q^*$-maximum base $B_0$ of $\M_1$ that contains $X^*$.
Since $B_0$ is a base of $\M_1$ containing $X^*$, we have $q^*(Z)\ge q^*(B_0)=\max_{B\in\cB_1}q^*(B)$.
Hence, $Z$ is a $q^*$-maximum base of $\M_1$.
We store the values $q^*(e)$ for $e\in Z$.

We next show how to test the condition in Line~\ref{line:cond} for the scanned element $e\in E\setminus X^*$ with $O(\log r)$ queries.
We maintain the invariant that $Z$ is a $q^*$-maximum base of $\M_1$ containing $Y$.
Note that $e\notin Y$ at this time.
Indeed, $Y=X^*\cup A'$, where $A'$ is the set of elements added to $Y$ in Line~\ref{line:buy} earlier in this round.
We have $e\notin X^*$, and $e\notin A'$ because $e$ has not been scanned before in this round.
We first compute $q(e)$ and $q^*(e)$ with $O(\log r)$ queries, and we may assume that $q(e)=q^*(e)>0$, since otherwise the condition does not hold.
If $e\in Z$, then $Y+e\subseteq Z$, and hence the condition holds.
Suppose that $e\notin Z$.
By Lemma~\ref{lem:test}, $Y+e$ is contained in a $q^*$-maximum base of $\M_1$ if and only if there exists $e'\in Z\setminus Y$ such that $Z-e'+e\in\cI_1$ and $q^*(e')=q^*(e)$.
Since $Z$ is a $q^*$-maximum base, every $e'\in Z\setminus Y$ with $Z-e'+e\in\cI_1$ satisfies $q^*(e')\ge q^*(e)$.
Hence, it suffices to find $e'\in Z\setminus Y$ minimizing $q^*(e')$ subject to $Z-e'+e\in\cI_1$ and to check whether $q^*(e')=q^*(e)$.
Such $e'$ can be found by Lemma~\ref{lem:bs} with $S=Z$, $R=Z\setminus Y$, $v=e$, and $w=q^*$ using $O(\log|Z|)=O(\log r)$ queries, where $|Z|=r_1(E)\le2r$.
One more query checks whether the element $e'$ found satisfies $Z-e'+e\in\cI_1$, and if it does not, then no such $e'$ exists and the condition does not hold.
If the condition holds and $e\notin Z$, we replace $Z$ by $Z-e'+e$, which is again a $q^*$-maximum base of $\M_1$ and contains $Y+e$.
Hence, the invariant is maintained.

Finally, Lines~\ref{line:pi}--\ref{line:buyend} require the value $p(e)$ and, if $X+e\notin\cI_2$, the element $f$.
Both were obtained when $q(e)$ was computed for the current $X$ and $\pi$, and hence no further queries are needed.

Since a round scans at most $n$ elements, each with $O(\log r)$ queries, a round uses $O(n\log r)$ queries in total.
\end{proof}

\begin{proof}[Proof of Theorem~\ref{thm:main}]
We run Algorithm~\ref{alg:weighted} with $\eps/3$ in place of $\eps$.
By Proposition~\ref{prop:approx}, it returns a $(1-\eps)$-approximate solution after $T=O(\eps^{-1}\log (r_1(E))+\eps^{-2})=O(\eps^{-1}\log r+\eps^{-2})$ rounds, where we used $r_1(E)\le 2r$.
By Lemma~\ref{lem:query}, the total number of independence queries is $O(n\log r\cdot(\eps^{-1}\log r+\eps^{-2}))$, plus $O(n)$ for the preprocessing in Section~\ref{sec:prelim} and for computing $r_1(E)$.
\end{proof}

\section*{Acknowledgments}

Our algorithm is entirely based on the unweighted matroid intersection algorithm of Huang--Kobayashi~\cite{huang2026auction}.
We are deeply grateful to Yusuke Kobayashi for sharing their result with us and for his many helpful comments that improved the clarity
and presentation of this paper.
We plan to combine this work with that of Huang--Kobayashi~\cite{huang2026auction} in a future version, providing a more polished and unified presentation.

This work was partially supported by the joint project of Kyoto University and Toyota Motor Corporation, titled ``Advanced Mathematical Science for Mobility Society'' and by JSPS KAKENHI Grant Number JP24KJ1494.

\section*{Acknowledgments of the use of AI}

Through discussions with GPT-5.6 Sol (OpenAI), we obtained an auction algorithm for weighted matroid intersection by combining the auction algorithm of Huang--Kobayashi~\cite{huang2026auction} for unweighted matroid intersection with the analysis of Liu--Ke--Khuller~\cite[Algorithm 1]{liu_et_al:LIPIcs.APPROX/RANDOM.2023.28} for weighted bipartite matching.
The latter analysis extends that of the auction algorithm of Assadi--Liu--Tarjan~\cite{assadi2021auction} for unweighted bipartite matching to the weighted setting.
Since the resulting analysis was hard to understand intuitively, we redid it using GPT-6 Astra (OpenAI).
We also used GPT-6 Astra and Claude Fable 5.1 (Anthropic) to assist in drafting the manuscript and improving its language and presentation.
We carefully verified and rewrote the proofs and take full responsibility for the paper.

\bibliography{biblio}

\appendix

\end{document}